\documentclass[
 aip,
 jmp,
 amsmath,amssymb,
 author-numerical,
 reprint,
 a4paper]{revtex4-1}

\usepackage{amsthm}
\usepackage{latexsym}
\usepackage[dvips]{epsfig}
\usepackage{mathrsfs}
\usepackage{yhmath}

\theoremstyle{plain}
\newtheorem{theorem}{Theorem}
\newtheorem{proposition}[theorem]{Proposition}
\newtheorem{lemma}[theorem]{Lemma}

\newtheorem{assumption}[theorem]{Assumption}

\newtheorem{definition}[theorem]{Definition}

\font\SYM=msbm10

\newcommand{\Complex}{\mbox{\SYM C}}

\begin{document}

\title{Constructing ``non-Kerrness'' on compact domains}

\author{Thomas B\"ackdahl}
\email{thomas.backdahl@aei.mpg.de}
\affiliation{ 
Max Planck Institut f\"ur Gravitationsphysik,  Albert Einstein Institut,
Am M\"uhlenberg 1, 14476 Golm, Germany
}
\author{Juan A. Valiente Kroon}
\email{j.a.valiente-kroon@qmul.ac.uk}
\affiliation{ 
School of Mathematical Sciences,  Queen Mary, University of London, 
Mile End Road, London E1 4NS, United Kingdom
}

\date{\today}

\begin{abstract}
Given a compact domain of a 3-dimensional hypersurface on a vacuum
spacetime, a scalar (the ``non-Kerrness'') is constructed by solving a
Dirichlet problem for a second order elliptic system. If such scalar
vanishes, and a set of conditions are satisfied at a point, 
then the domain of dependence of the compact domain is
locally isometric to a portion of a member of the Kerr family of solutions to
the Einstein field equations. This construction is expected to be of
relevance in the analysis of numerical simulations of black hole
spacetimes.
\end{abstract}


\maketitle

\section{Introduction}
The present article is concerned with the problem of measuring how
different a given initial data set for the Einstein vacuum field
equations is from a Kerr initial data set. In
\cite{BaeVal10a,BaeVal10b,BaeVal11b} this problem has been addressed
by the construction of a geometric invariant ---the
\emph{non-Kerrness}--- on hypersurfaces with at least one asymptotic
end. This setting, although convenient for theoretical discussions, is
not ideal for numerical considerations where very often one needs to
make use of bounded computational domains on an hypersurface. The
purpose of this article is to provide a construction of non-Kerrness
on bounded domains.

\medskip
The construction of the non-Kerrness given in
\cite{BaeVal10a,BaeVal10b,BaeVal11b} is based on a very strong
property of the Kerr spacetime: the existence of a \emph{Killing-Yano
tensor}. A Killing-Yano tensor is an antisymmetric, rank 2 tensor
$Y_{\mu\nu}$ satisfying the equation
\begin{equation*}
\nabla_{(\mu} Y_{\nu)\lambda}=0.
\end{equation*}
Let $\zeta_{\mu}\equiv \epsilon_\mu{}^{\nu\lambda\rho} \nabla_\nu
Y_{\lambda \rho}$ denote the codifferential of $Y_{\mu\nu}$. If
$Y_{\mu\nu}$ is a Killing-Yano tensor, then $\zeta_\mu$ satisfies the
Killing vector equation. As discussed in \cite{PenRin86}, the theory
of Killing-Yano tensors can be conveniently reformulated in terms of
the existence of a valence 2 Killing spinor,
$\kappa_{AB}=\kappa_{(AB)}$, satisfying the equation
\begin{equation}
\nabla_{A'(A} \kappa_{BC)}=0.
\label{KillingSpinorEquation}
\end{equation}
The spinorial analogue of the codifferential $\zeta_\mu$ is the spinor
$\xi_{AA'}\equiv \nabla_{A'}{}^B\kappa_{AB}$. In general, if
$\kappa_{AB}$ satisfies the Killing spinor equation, then $\xi_{AA'}$
is a complex Killing vector. In the case of the Kerr spacetime the
real and imaginary parts of this vector are proportional ---and 
by multiplying with a complex constant, the imaginary part can be set to zero. 
In general, the existence of a Killing-Yano tensor is equivalent to existence of 
a Killing spinor $\kappa_{AB}$ such that $\xi_{AA'}$ is real.

\medskip
Killing spinors (or alternatively, Killing-Yano tensors) are useful in
the characterisation of the Kerr spacetime as the existence of one of
these objects severely restricts the algebraic type of the curvature
of the spacetime. Furthermore, the implied existence of a real Killing
vector allows to make contact with the theory of the Mars-Simon tensor
---see \cite{Mar99,Mar00}. As a result of this analysis, it is
possible to provide a purely local
characterisation of the Kerr spacetime ---see Theorem~1 in \cite{Mar00}. Alternatively, one can obtain a somewhat simpler characterisation
if one combines local and global requirements: the existence of a
stationary, asymptotically flat region with non-vanishing
mass ---see Theorem~2
in \cite{Mar00}. Precisely this result was used in the constructions of
non-Kerrness on non-bounded 3-manifolds described in
\cite{BaeVal10a,BaeVal10b,BaeVal11b}.

\medskip
The construction of the non-Kerrness on bounded domains 
discussed in the present article makes use of the local spacetime
characterisation of the Kerr spacetime given in Theorem~1 of
\cite{Mar00} to show that if the non-Kerrness vanishes on some
3-dimensional bounded domain, then the initial data prescribed on that
region is locally isometric to data for a Kerr spacetime. We expect
that this result will be of utility to assess in a quantitative way
how a given numerically constructed dynamical black hole spacetime
evolves towards a stationary state described by the Kerr spacetime. In
the process, it will be shown that the general theory of Killing
spinor initial data sets used in \cite{BaeVal10a,BaeVal10b,BaeVal11b}
can be simplified.

\subsection*{Overview of the article}

The content of this article is structured as follows: Section
\ref{Section:KillingSpinors} provides a summary of key properties of
spacetimes with Killing spinors. It also contains a reformulation in
terms of spinors of a local characterisation of the Kerr spacetime by
M. Mars. Finally, a brief discussion of the notion of Killing spinor
candidates is provided. Section~\ref{Section:KillingSpinorInitialData}
provides a brief summary of the theory of the Killing spinor initial
data equations which encode the existence of a Killing vector at the
level of initial data. Section~\ref{Section:ApproximateKillingSpinors}
gives a brief discussion of the notion of approximate Killing spinors,
the approximate Killing spinor equations and the elliptic theory
required to discuss the existence of solutions to this equation with
Dirichlet boundary conditions. Section~\ref{Section:KillingVectors}
provides a result regarding the realness of the Killing vector
constructed from the Killing spinor, which will be required in our
subsequent discussion.  Section~\ref{Section:MainResult} provides our
main result: a theorem which characterises Kerr initial data on a
compact domain of a 3-dimensional manifold using the notion of
approximate Killing spinors. Finally Section~\ref{Section:Conclusions}
provides some concluding remarks. There is an appendix
(Appendix~\ref{RedundancySecondAlgebraicCondition}) providing a proof
of a theorem discussed in
Section~\ref{Section:KillingSpinorInitialData}, which tells that one
of the Killing spinor initial data equations can be omitted.

\subsection*{Notation and conventions}
All throughout, $(\mathcal{M},g_{\mu\nu})$ will denote a smooth, orientable and
time orientable globally hyperbolic vacuum spacetime. Here, and in
what follows, $\mu,\,\nu,\cdots$ denote abstract 4-dimensional tensor
indices. The metric $g_{\mu\nu}$ will be taken to have signature
$(+,-,-,-)$. Let $\nabla_\mu$ denote the Levi-Civita connection of
$g_{\mu\nu}$. The sign of the Riemann tensor will be given by the
equation
\[
\nabla_\mu\nabla_\nu\xi_\zeta-\nabla_\nu\nabla_\mu\xi_\zeta=R_{\nu\mu\zeta}{}^\eta\xi_\eta.
\]
Spinors will be used systematically. We follow the conventions of
\cite{PenRin84}.  In particular, $A,\,B,\ldots$ will denote abstract
spinorial indices. Tensors and their spinorial
counterparts are related by means of the solder form $\sigma_\mu{}^{AA'}$ satisfying $g_{\mu\nu}=\sigma_\mu^{AA'}\sigma_\nu^{BB'} \epsilon_{AB}\bar\epsilon_{A'B'}$, where $\epsilon_{AB}$
is the antisymmetric spinor and $\bar\epsilon_{A'B'}$ its complex
conjugate copy. One has, for example, that $\xi_\mu =
\sigma_{\mu}{}^{AA'} \xi_{AA'}$.  Let $\nabla_{AA'}$ denote the
spinorial counterpart of the spacetime connection
$\nabla_\mu$.

\section{A local spacetime characterisation of the Kerr spacetime}
\label{Section:KillingSpinors}

Given a spacetime $(\mathcal{M},g_{\nu\nu})$, let
$C_{\mu\nu\lambda\rho}$ denote the Weyl tensor of the metric
$g_{\mu\nu}$. Let $C_{AA'BB'CC'DD'}$ denote the spinorial counterpart of
$C_{\mu\nu\lambda\rho}$. There exists a completely symmetric spinor $\Psi_{ABCD}$ such that:
\begin{align*}
C_{AA'BB'CC'DD'} ={}& \Psi_{ABCD} \bar{\epsilon}_{A'B'}\bar{\epsilon}_{C'D'} \\
&+ \bar{\Psi}_{A'B'C'D'} \epsilon_{AB}\epsilon_{CD}.
\end{align*}
In terms of the spinor $\Psi_{ABCD}$, the Bianchi identity can be
rewritten as
\begin{equation}
\nabla^Q{}_{A'}\Psi_{ABCQ}=0.
\label{SpinorialBianchiIdentity}
\end{equation}

\medskip
We recall that the two classical invariants of the Weyl tensor are
given by:
\begin{align*}
\mathcal{I} &\equiv \tfrac{1}{2} \Psi_{ABCD}\Psi^{ABCD}, \\
\mathcal{J} &\equiv \tfrac{1}{6} \Psi_{ABCD} \Psi^{CDEF}\Psi_{EF}{}^{AB}.
\end{align*}

\subsection{Properties of spacetimes with Killing spinors}
\label{Subsection:PropertiesKillingSpinors}
In what follows it is assumed one has a region $\mathcal{N}$ of the spacetime
$(\mathcal{M},g_{\mu\nu})$ where one has a solution $\kappa_{AB}$ of
the Killing spinor equation \eqref{KillingSpinorEquation}. It is then
well known that the spacetime must be of Petrov type D, N or O at every point where the Killing spinor exists ---see
e.g. \cite{WalPen70}. In the sequel we will concentrate our
attention to the case when $(\mathcal{M},g_{\mu\nu})$ is of Petrov
type D. In such case, there exist spinors $\alpha_A$,  $\beta_A$,
$\alpha_Q \beta^Q=1$, such
that
\begin{equation}
\label{TypeD:Psi}
\Psi_{ABCD} =-\psi \alpha_{(A} \alpha_A \beta_C \beta_{D)},
\end{equation}
where
\begin{equation}
\psi \equiv 18 \mathcal{J}/\mathcal{I}. 
\label{psi}
\end{equation}
The sign convention used in this equation differs from the one used in
\cite{BaeVal10a,BaeVal10b,BaeVal11b}. The reason behind this choice is
to avoid potential problems with the choice of branch of roots of
complex quantities. The valence 2 Killing spinor is then given by
\begin{equation}
\label{TypeD:kappa}
\kappa_{AB} = \psi^{-1/3} \alpha_{(A}\beta_{B)},
\end{equation}
where the branch with minimal absolute value of the complex argument
is used. The conventions used gives a real and positive $\psi$ for the
Schwarzschild spacetime. 

\medskip
As in the
introduction, let 
\[
\xi_{AA'}\equiv \nabla^Q{}_{A'} \kappa_{AQ}.
\]
Then $\xi_{AA'}$ is (in general) a complex solution to Killing equation
\[
\nabla_{AA'}\xi_{BB'} + \nabla_{BB'} \xi_{AA'}=0. 
\]
If $\xi_{AA'}$ is real, we define the Killing form of $\xi_{AA'}$ by 
\begin{align*}
F_{AA'BB'} &\equiv \tfrac{1}{2} \left(\nabla_{AA'} \xi_{BB'} - \nabla_{BB'}
\xi_{AA'}\right)\\
&=\nabla_{AA'} \xi_{BB'}. 
\end{align*}
Vacuum spacetimes admitting a Killing spinor such that
$\xi_{AA'}$ is real will be said to belong to the \emph{generalised
Kerr-NUT class} ---see \cite{BaeVal10a,BaeVal10b}. \emph{In the rest of this section it is assumed that
$(\mathcal{M},g_{\mu\nu})$ is a generalised Kerr-NUT spacetime.}

\medskip
As a consequence of the symmetries of $F_{AA'BB'}$,
there exists a symmetric, valence 2 spinor $\phi_{AB}$ such that
\begin{align*}
F_{AA'BB'} &= \phi_{AB} \bar{\epsilon}_{A'B'} + \bar{\phi}_{A'B'} \epsilon_{AB}, \\
\phi_{AB} &\equiv \tfrac{1}{2} F_{AQ'B}{}^{Q'}. 
\end{align*}
Using  \eqref{TypeD:kappa} one finds the following expressions for
$\xi_{AA'}$, and $\phi_{AB}$ in terms of $\psi$ and the principal
spinors:
\begin{align*}
\xi_{AA'} &= \tfrac{1}{2}\psi^{-4/3}\left(  \alpha_A \beta^Q +
  \beta_A \alpha^Q \right) \nabla_{QA'}\psi, \\
\phi_{AB} &=-\tfrac{3}{4}\Psi_{ABCD}\kappa^{CD}= -\tfrac{1}{4} \psi^{2/3} \alpha_{(A} \beta_{B)}. 
\end{align*}
The above expression for the spinor $\phi_{AB}$ is obtained using the
Killing spinor equation and by
commutation of covariant derivatives.

\medskip
For later use, we introduce the \emph{norm of the Killing form}, the
\emph{norm of the Killing vector} and the \emph{twist 1-form} via
\begin{align*}
\Phi &\equiv \phi_{PQ} \phi^{PQ}, \qquad \lambda \equiv \xi_{AA'}\xi^{AA'},\\
\omega_{AA'} &\equiv \epsilon_{AA'BB'CC'DD'} \xi^{BB'}\nabla^{CC'}\xi^{DD'}, 
\end{align*}
where 
\begin{align*}
\epsilon_{AA'BB'CC'DD'} \equiv{}& \mbox{i} \left(
  \epsilon_{AC} \epsilon_{BD} \bar\epsilon_{A'D'} \bar\epsilon_{B'C'}\right. \\
&\left. - \epsilon_{AD} \epsilon_{BC} \bar\epsilon_{A'C'} \bar\epsilon_{B'D'}
\right)
\end{align*}
is the spinorial counterpart of the completely antisymmetric volume
form, $\epsilon_{\mu\nu\lambda\rho}$, of $g_{\mu\nu}$. Locally,
$\omega_{AA'}$ is exact, so that there exists $\omega$ (the
\emph{twist potential}) such that $\omega_{AA'} = \nabla_{AA'}
\omega$. Using $\lambda$ and $\omega$ we define the \emph{Ernst
potential}, $\sigma$, by
\[
\sigma \equiv \lambda + \mbox{i} \omega.
\] 

\medskip
Using expressions \eqref{TypeD:Psi} and \eqref{TypeD:kappa} one
readily finds the following expressions for $\Phi$, $\lambda$ and
$\omega_{AA'}$:
\begin{subequations}
\begin{eqnarray}
&& \Phi= -\tfrac{1}{32} \psi^{4/3},  \label{TypeD:Phi}\\
&& \lambda= -\tfrac{1}{4} \psi^{-8/3} \nabla_{AA'}\psi
\nabla^{AA'}\psi, \label{TypeD:lambda}\\
&& \omega_{AA'} =\mbox{Im}(4\phi_{A}{}^{B}\xi_{BA'}), \label{TypeD:twist1-form} 
\end{eqnarray}
\end{subequations}
In order to obtain an expression for the Ernst potential in terms of
$\psi$, we notice the identities
\begin{subequations}
\begin{eqnarray}
&& \nabla_{AA'}(\psi^{1/3}) =-\tfrac{16}{3}\phi_{A}{}^{B}\xi_{BA'}
, \label{Concomitant1}\\
&& \nabla_{AA'}\lambda =\mbox{Re}(4\phi_{A}{}^{B}\xi_{BA'}). \label{Concomitant3}
\end{eqnarray}
\end{subequations}
These identities follow from the Bianchi identity
\eqref{SpinorialBianchiIdentity}, the Killing spinor equation and
commuting derivatives as nessesary. One concludes that
\[
\nabla_{AA'}\lambda + \mbox{i} \omega_{AA'} = -\tfrac{3}{4}\nabla_{AA'}\psi^{1/3}.
\]
The latter can be integrated to give
\begin{equation}
\label{ErnstPotential:Psi}
\sigma - c = -\tfrac{3}{4}\psi^{1/3},
\end{equation}
with $c$ a complex constant. The real part of $c$ is not arbitrary:
using equations \eqref{Concomitant1} and \eqref{Concomitant3} one
obtains that
\begin{equation}
\label{Rec}
\mbox{Re}(c)=\lambda + \tfrac{3}{4}\mbox{Re}( \psi^{1/3}).
\end{equation}

\subsection{A local characterisation of Kerr}

The analysis of the \emph{so called} Mars-Simon tensor presented in
\cite{Mar99,Mar00} gives rise to a local characterisation of the Kerr
spacetime among the class of spacetimes endowed with a Killing
vector. This characterisation involves the Weyl tensor, the Killing
form and the Ernst potential ---see Theorem~1 in \cite{Mar00}. For the
convenience of our subsequent analysis, here we present a slight
generalisation of this result in the language of spinors.

\begin{theorem}[Mars, 2000]
\label{Theorem:Mars2000}
Let $(\mathcal{M},g_{\mu\nu})$ be a smooth, vacuum spacetime admitting
a Killing vector $\xi^\mu$. Let $\mathcal{N}\subset\mathcal{M}$ be a
non-empty open subset satisfying:

\begin{itemize}
\item[(i)] There is a point $p\in \mathcal{N}$
where $\Phi\neq 0$.

\item[(ii)] The Killing form and the Weyl tensor are related by
\[
\Psi_{ABCD} = \varpi \phi_{(AB}\phi_{CD)},
\]
where $\varpi$ is a complex scalar function.
\end{itemize}

\noindent
Then there exist two complex constants $\tilde{c}$ and $k$ such that
\[
\varpi=-\frac{12}{\tilde{c}-\sigma} , \qquad \Phi =-k(\tilde{c}-\sigma)^4, \qquad \mbox{ on } \mathcal{N}.
\]
If, in addition, $\mbox{\em Re}(\tilde{c})>0$ and $k=\mbox{\em
Re}(k)>0$ then $(\mathcal{N},g_{\mu\nu})$ is locally isometric to a
portion of the Kerr spacetime. 

\end{theorem}

\medskip
\noindent
\textbf{Remark 1.} This result follows from ---and is equivalent to---
Theorem~1 in \cite{Mar00} by introducing a different normalisation in
the Killing vector and exploiting the fact that $\omega$ is defined only up to an additive constant. 
We thank M. Mars for pointing this out to us. 

\medskip
\noindent
\textbf{Remark 2.} As discussed in \cite{Mar00} it follows from the
previous result that the
Kerr spacetime is everywhere strictly of type D. In particular this
implies that 
$\psi\neq 0$.

\subsection{Killing spinor candidates}

The construction of non-Kerrness on  a bounded domain requires the
notion of a \emph{Killing spinor candidate} introduced in \cite{BaeVal11b}:

\begin{definition}
\label{Definition:KSCandidate}
Let $(\mathcal{M},g_{\mu\nu})$ be a vacuum spacetime. Consider 
a point $p\in \mathcal{M}$ for which $\mathcal{I}\neq 0$, $\mathcal{J}\neq 0$ and a symmetric spinor
$\zeta_{AB}$  satisfying at $p$
\[
\zeta_{AB}\neq 0, \;\;
 \psi^{-1}\Psi_{PQRS} \zeta^{PQ}\zeta^{RS} + \tfrac{1}{6} \zeta_{PQ}\zeta^{PQ}\neq 0.
\]
The symmetric spinor given by
\begin{equation}
\breve{\kappa}_{AB} = \psi^{-1/3} \Xi^{-1/2}
\negthickspace\left(-\psi^{-1}\Psi_{ABPQ} \zeta^{PQ}\negthinspace -\tfrac{1}{6}\zeta_{AB}
\right)\negthickspace, \label{Candidate}
\end{equation}
with
\[
 \Xi \equiv -\psi^{-1}\Psi_{PQRS} \zeta^{PQ}\zeta^{RS} - \tfrac{1}{6}
\zeta_{PQ}\zeta^{PQ}, 
\]
will be called the $\zeta_{AB}$-Killing spinor candidate at
 $p$.The scalar $\psi$ is obtained from the Weyl spinor
$\Psi_{ABCD}$ using formula \eqref{psi}. 
\end{definition}

Formula \eqref{Candidate} can be evaluated for any vacuum
spacetime $(\mathcal{M},g_{\mu\nu})$ satisfying the explicit conditions in definitin~\ref{Definition:KSCandidate}, that is it is not restricted to a special Petrov type. 
The name Killing spinor candidate is justified by the following result also proved in
\cite{BaeVal11b}:

\begin{proposition}
Let $(\mathcal{M},g_{\mu\nu})$ be a vacuum spacetime. If on
$\mathcal{N}\subset \mathcal{M}$,  the spacetime is of Petrov type D and
$\zeta_{AB}$ is a symmetric spinor satisfying
\begin{align*}
&\Xi=\psi^{-1}\Psi_{PQRS} \zeta^{PQ}\zeta^{RS} + \tfrac{1}{6} \zeta_{PQ}\zeta^{PQ}\neq 0, \\
&\zeta_{AB}\neq 0 \quad \mbox{ on } \mathcal{N},
\end{align*}
and $\mathcal{N}$ contains no branch cuts of $\psi^{1/3}$ and
$\Xi^{1/2}$, then
\begin{equation}
\kappa_{AB} = \psi^{-1/3} \Xi^{-1/2}\negthickspace \left(-\psi^{-1}\Psi_{ABPQ} \zeta^{PQ}\negthinspace -\tfrac{1}{6}\zeta_{AB}  \right) \label{KillingSpinorFormula}
\end{equation}
is a Killing spinor on $\mathcal{N}$. The formula \eqref{KillingSpinorFormula} is independent
of the choice of $\zeta_{AB}$.
\end{proposition}

\medskip
\noindent
{\bf Remark.} Different choices of branch cuts in $\psi^{1/3}$ and
$\Xi^{1/2}$ only change the right hand side of
\eqref{KillingSpinorFormula} by a constant complex phase. The
assumption on the no existence of branch cuts of $\psi^{1/3}$ and
$\Xi^{1/2}$ is included to ensure the smooth existence of derivatives
of the various fields ---see also Assumption~\ref{AssumptionXi} below.

\section{The Killing spinor initial data equations}
\label{Section:KillingSpinorInitialData}

Key for the construction of the non-Kerrness discussed in
\cite{BaeVal10a,BaeVal10b,BaeVal11b}, is the idea of how to encode
that the development of an initial data set
$(\mathcal{S},h_{ij},K_{ij})$ admits a solution to the Killing spinor
equation \eqref{KillingSpinorEquation}. This question can be addressed
by means of the space-spinor decomposition of the Killing spinor
equation \eqref{KillingSpinorEquation}. For a more detailed description see~\cite{BaeVal10b}.

\medskip
In order to perform a space-spinor decomposition of equation
\eqref{KillingSpinorEquation} it is convenient to define the spinors
\begin{align}
\xi_{ABCD} &\equiv \nabla_{(AB} \kappa_{CD)}, \quad \xi_{AB} \equiv
\tfrac{3}{2}\nabla_{(A}{}^{D}\kappa_{B)D},\nonumber\\
& \xi \equiv \nabla^{PQ}\kappa_{PQ},
\label{Definition:xi}
\end{align}
where $\nabla_{AB}$ denotes the spinorial version of the Sen connection
associated to the pair $(h_{ij},K_{ij})$ of intrinsic metric and
extrinsic curvature. It can be expressed in terms of the spinorial
counterpart, $D_{AB}$ of the Levi-Civita connection of the 3-metric
$h_{ij}$, and the spinorial version, $K_{ABCD}=K_{(AB)(CD)}=K_{CDAB}$,
of the second fundamental form $K_{ij}$. For example, given a valence
1 spinor $\pi_{A}$ one has that
\[
\nabla_{AB}\pi_C = D_{AB} \pi_C +\tfrac{1}{2} K_{ABC}{}^Q \pi_Q,
\]  
with the obvious generalisations to higher valence spinors. For expression of the commutators we refer to the paper~\cite{BaeVal10b}. The
\emph{Hermitian conjugate} of $\pi_A$ is defined via
\[
\hat{\pi}_{A} \equiv \tau_{A}{}^{E'}\bar{\pi}_{E'},
\] 
where $\tau^{AA'}$ is the normal to $\mathcal{S}$ with length $\sqrt{2}$.
The Hermitian conjugate can be extended to higher valence symmetric spinors in the
obvious way.  It can be verified that $\xi_{ABCD} \hat{\xi}^{ABCD}\geq 0$.

\medskip
Using the notation described in the previous paragraph we find that the space-spinor decomposition of equation
\eqref{KillingSpinorEquation} renders a set of 3 conditions intrinsic
to the hypersurface $\mathcal{S}$:
\begin{subequations}
\begin{align}
& \xi_{ABCD}=0,\label{kspd1}\\
& \Psi_{(ABC}{}^F\kappa_{D)F}=0, \label{kspd2}\\
& 3\kappa_{(A}{}^E\nabla_B{}^F\Psi_{CD)EF}+\Psi_{(ABC}{}^F\xi_{D)F}=0,\label{kspd3}
\end{align}
\end{subequations}
where the spinor $\Psi_{ABCD}$
denotes, in a slight abuse of notation, the restriction to the
hypersurface $\mathcal{S}$ of the self-dual Weyl spinor. For the ease
of notation, a similar
convention will be adopted for the restriction of other spacetime
fields. Whether one is considering the field on spacetime or its
restriction to $\mathcal{S}$ will always be clear from the
context. Crucially, the spinor $\Psi_{ABCD}$ in equations \eqref{kspd2}-\eqref{kspd3} can be written
entirely in terms of initial data quantities via the relations:
\[
\Psi_{ABCD} = E_{ABCD} + \mbox{i} B_{ABCD},
\]
with
\begin{align*}
E_{ABCD}={}& -r_{(ABCD)} + \tfrac{1}{2}\Omega_{(AB}{}^{PQ}\Omega_{CD)PQ}\\
&- \tfrac{1}{6}\Omega_{ABCD}K, \\
B_{ABCD}={}&-\mbox{i}\ D^Q{}_{(A}\Omega_{BCD)Q},
\end{align*}
and where $\Omega_{ABCD}\equiv K_{(ABCD)}$, $K\equiv
K_{PQ}{}^{PQ}$. Furthermore, the spinor $r_{ABCD}$ is the Ricci
tensor, $r_{ij}$, of the 3-metric $h_{ij}$.

\medskip
In Appendix \ref{RedundancySecondAlgebraicCondition} it is shown that
the second algebraic condition \eqref{kspd3} is, in fact, redundant and
a consequence of the conditions \eqref{kspd1}-\eqref{kspd2}. In
particular it follows then that 

\begin{theorem}
\label{Theorem:kspd}
Let equations \eqref{kspd1}-\eqref{kspd2} be satisfied for a symmetric
spinor $\check{\kappa}_{AB}$ on an open set $\mathcal{U}\subset
\mathcal{S}$. Then the Killing spinor equation
\eqref{KillingSpinorEquation} has a solution, $\kappa_{AB}$, on the
future domain of dependence $\mathcal{D}^+(\mathcal{U})$.
\end{theorem}

\medskip
\noindent
\textbf{Remark.} This means that the term $I_2$ in the invariants of
\cite{BaeVal10a,BaeVal10b,BaeVal11b} can be omitted.

\section{Approximate Killing spinors}
\label{Section:ApproximateKillingSpinors}

\subsection{The approximate Killing spinor equation}
The \emph{spatial Killing spinor equation} \eqref{kspd1} can be
regarded as a (complex) generalisation of the conformal Killing vector 
equation.  As in the case of the conformal Killing equation, equation \eqref{kspd1} is
clearly overdetermined. However, one can construct a generalisation of
the equation which under suitable circumstances can always be expected
to have a solution.  One can do this by composing the operator in \eqref{kspd1} with
its formal adjoint ---see \cite{BaeVal10a}. This procedure renders the
equation
\begin{align}
\mathbf{L}\kappa_{CD} \equiv{}& \nabla^{AB} \nabla_{(AB} \kappa_{CD)}
-\Omega^{ABF}{}_{(A}\nabla_{|DF|}\kappa_{B)C}\nonumber\\
&-\Omega^{ABF}{}_{(A}\nabla_{B)F}\kappa_{CD}=0, \label{ApproximateKillingSpinorEquation} 
\end{align}
which will be called the \emph{approximate Killing spinor
  equation}. One has the following result proved in \cite{BaeVal10b}:

\begin{lemma}
The operator $\mathbf{L}$ defined by the left hand side of equation
\eqref{ApproximateKillingSpinorEquation} is a formally self-adjoint elliptic
 operator.
\end{lemma}

In order to discuss the solvability of equation \eqref{ApproximateKillingSpinorEquation} 
on a bounded domain, $\mathcal{U}\subset \mathcal{S}$,
one has to supplement it with appropriate boundary conditions. 
On $\partial \mathcal{U}$ we will consider the homogeneous 
Dirichlet operator $\mathbf{B}$ given by
\[
\mathbf{B}u(y) =u(y), \quad y\in \partial \mathcal{S}.
\]
The combined operator $(\mathbf{L},\mathbf{B})$ satisfies the
so-called \emph{Lopatinski-Shapiro compatibility conditions} ---see
\cite{WloRowLaw95} for detailed definitions and discussion. Thus, $(\mathbf{L},\mathbf{B})$ is L-elliptic ---see again
\cite{WloRowLaw95}, Theorem~10.7. Moreover, one has the following
theorem ---see also \cite{Nir55}. 

\begin{theorem}
\label{Fredholm:properties}
Let $\mathbf{L}$ denote a smooth second order homogeneous elliptic
operator on $\mathcal{U}$. Furthermore, let $\partial \mathcal{U}$ be
smooth and let $\mathbf{B}$ denote the Dirichlet boundary
operator. Then for $s\geq 2$ the map
\[
(\mathbf{L},\mathbf{B}): H^s(\mathcal{U}) \rightarrow H^{s-2}(\mathcal{U}) \times H^{s-1/2}(\partial \mathcal{U}) 
\]
is Fredholm. Furthermore, the boundary value problem
\begin{eqnarray*}
&& \mathbf{L} u(x) = f(x), \quad f \in H^0(\mathcal{U}), \quad x\in
\mathcal{U}, \\
&& u(y)=g(y), \quad g \in H^0(\partial\mathcal{U}), \quad y
\in \partial \mathcal{U},
\end{eqnarray*}
has a solution $u\in H^2(\mathcal{U})$ if
\[
\int_{\mathcal{U}} f\cdot \nu \mbox{d}\mu =0, 
\]
for all $\nu \in H^2(\mathcal{U})$ such that
\begin{eqnarray*}
&& \mathbf{L}^*\nu(x)=0, \quad x\in \mathcal{U}, \\
&& \nu(y) =0, \quad y\in \partial \mathcal{U}.
\end{eqnarray*}
\end{theorem}

\medskip
\noindent
\textbf{Remark 1.} In the previous Theorem, the action of
$\mathcal{B}$ on $u$ is to be understood in the trace sense ---see
\cite{WloRowLaw95}. 

\medskip
\noindent
\textbf{Remark 2.} If $\mathbf{L}$ has smooth coefficients and
$\mathbf{L}u=0$, then it follows from Weyl's Lemma ---see e.g.
\cite{WloRowLaw95}--- that if a solution to the boundary value problem
exists and the boundary data is smooth, then the solution must be, in
fact, smooth ---this is the so-called elliptic regularity.

\medskip
In what follows let $n_{AB}=n_{(AB)}$ denote the spinorial counterpart
of the inward pointing normal to $\partial \mathcal{U}$. As a
consequence of our signature conventions one has that $n_{PQ}
n^{PQ}=-1$. Theorem~\ref{Fredholm:properties} will be used to establish the
existence of solutions to the approximate Killing spinor equation
\eqref{ApproximateKillingSpinorEquation} with Dirichlet boundary data
given by the $n_{AB}$-Killing spinor candidate. In order to ensure
that the Killing spinor candidate can be constructed on $\partial
\mathcal{U}$, we define the set
\[
\mathcal{Q} \equiv \{ z \in \Complex \; | \; z=\Xi(p), \;\; p
\in \partial \mathcal{U}\},
\]
where we have chosen $\zeta_{AB}=n_{AB}$ in the function $\Xi$. We
make the following assumption:

\begin{assumption}
\label{AssumptionXi}
The initial data set $(\mathcal{S},h_{ij},K_{ij})$ and the compact
set $\mathcal{U}$ are such that $\mathcal{I}\neq 0$, $\mathcal{J}\neq 0$ on $\partial \mathcal{U}$ and that $\Xi$ is a smooth function over $\partial \mathcal{U}$ satisfying
\begin{itemize}
\item[(i)] $0\not \in \mathcal{Q}$;
\item[(ii)] $\mathcal{Q}$ does not encircle the point $z=0$.
\end{itemize}
when we choose $\zeta_{AB}$ as the inward pointing normal to $\partial \mathcal{U}$. 
\end{assumption}

\medskip
\noindent
\textbf{Remarks.} As a consequence of this assumption one can choose a cut of the square
root function on the complex plane such that $\Xi^{1/2}(p)$ is smooth
for all $p\in \partial \mathcal{U}$. Notice that the $n_{AB}$-Killing
spinor candidate is only defined at $\partial\mathcal{U}$. The
assumptions  $\mathcal{I}\neq 0$, $\mathcal{J}\neq 0$ are justified on
the basis that we are mainly interested in discussing configurations
close to Kerr initial data ---for which $\psi\neq 0$.

\medskip
One has the following result:

\begin{proposition}
\label{Existence:ApproximateKS}
Let $(\mathcal{S},h_{ij},K_{ij})$ be an initial data set for the
Einstein vacuum field equations. Furthermore, let $\mathcal{U}\subset
\mathcal{S}$ be a compact subset with boundary $\partial \mathcal{S}$
satisfying Assumption~\ref{AssumptionXi}. Then, there exists a
unique smooth solution, $\kappa_{AB}$, to the approximate Killing spinor
equation \eqref{ApproximateKillingSpinorEquation} with boundary value
given by the $n_{AB}$-Killing spinor candidate given pointwise by equation
\eqref{Candidate} on $\partial \mathcal{U}$.
\end{proposition}

\begin{proof}
The proof of this result follows directly from the second part of
Theorem~\ref{Fredholm:properties}. Notice that as the equation is is
homogeneous, there is no potential obstruction to the existence of
solutions and one does not need to verify the triviality of the Kernel
of the adjoint operator as it is in the case with asymptotically
Euclidean ends ---see \cite{BaeVal10a,BaeVal10b,BaeVal11b}.
\end{proof}

\section{Reality of the Killing vector}
\label{Section:KillingVectors}
As discussed in the introduction, the existence of a Killing spinor is
not enough to single out the generalized Kerr-NUT family from the type
D solutions. We also need that the Killing vector constructed from the
Killing spinor is real. This section provides some tools to determine
that.

\subsection{Imaginary part of the Killing vector data}

In what follows, let $\kappa_{AB}$ solve the Killing spinor equation
\eqref{KillingSpinorEquation} in a spacetime domain $\mathcal{D}$, and
let $\xi$ and $\xi_{AB}$ be defined as in \eqref{Definition:xi}.  In
this section we only study what happes in the domain $\mathcal{D}$.  A
computation using the suite {\tt xAct} for {\tt Mathematica} starting
from equations \eqref{kspd1}-\eqref{kspd3} shows that
\begin{subequations}
\begin{align}
D_{AB}\mbox{Im}(\xi ^{AB}) ={}& -\tfrac{1}{2} \mbox{Im}(\xi)  K,\label{DImxi2a}\\
D_{(AB}\mbox{Im}(\xi _{CD)}) ={}& -\tfrac{1}{2} \mbox{Im}(\xi)  \Omega_{ABCD}.\label{DImxi2b}
\end{align}
\end{subequations}
This can be seen by using equations (18a) and (18b) in~\cite{BaeVal10b} and splitting into real and imaginary parts.
Equation \eqref{KillingSpinorEquation} implies
$\nabla\kappa_{AB}=-\tfrac{2}{3}\xi_{AB}$, where $\nabla$ denotes the
normal derivative $\tau^{AA'}\nabla_{AA'}$. Commuting derivatives and
simplifying one obtains
\begin{subequations}
\begin{align}
\nabla \mbox{Im}(\xi)  ={}& \mbox{Im}(\xi^{AB}) K_{AB},\label{nablaImxi2a}\\
\nabla \mbox{Im}(\xi _{AB}) ={}& -\tfrac{1}{2} \mbox{Im}(\xi)  K_{AB}
+ \tfrac{1}{3} \mbox{Im}(\xi _{AB}) K  \nonumber\\
&+ \Omega_{ABCD}\mbox{Im}(\xi^{CD})
-  D_{AB}\mbox{Im}(\xi) \nonumber\\ 
&-  \mbox{Im}(\xi_{(A}{}^{C})K_{B)C},\label{nablaImxi2b}
\end{align}
\end{subequations}
where $K_{AB}$ is the acceleration vector. For more details about the derivation see equations (32b) and (32c) in~\cite{BaeVal10b} and their derivations.
Making a space spinor split of $\xi_{AA'}=\nabla^B{}_{A'}\kappa_{AB}$
and using equation \eqref{KillingSpinorEquation}, we find
\[
\mbox{Im}(\xi_{AA'})=\tfrac{1}{2} \mbox{Im}(\xi) \tau_{AA'} -
\mbox{Im}(\xi _{AB}) \tau^{B}{}_{A'}.
\]
 After differentiating once more,
making a further space spinor split, and using equations \eqref{DImxi2a},
\eqref{DImxi2b}, \eqref{nablaImxi2a} and \eqref{nablaImxi2b} we have:
\begin{lemma}
\label{Lemma:KDTriviality}
Let $\kappa_{AB}$ solve the Killing spinor equation \eqref{KillingSpinorEquation} in a spacetime domain $\mathcal{D}$. Assume that 
\begin{align}
&\mbox{\em Im}(\xi)=0, \quad \mbox{\em Im}(\xi_{AB})=0, \quad
D_{AB}\mbox{\em Im} (\xi)=0,\nonumber\\
&D_{(A}{}^{C} \mbox{\em Im}(\xi_{B)C})=0
\end{align}
at a point $p\in \mathcal{D}$. Then $\mbox{\em Im}(\xi_{AA'})=0$ and
$\nabla_{AA'}\mbox{\em Im}(\xi_{BB'})=0$ at $p$.
\end{lemma}

\section{The non-Kerrness invariant}
\label{Section:MainResult}

The approximate Killing spinor $\kappa_{AB}$ obtained in Proposition~\ref{Existence:ApproximateKS} 
will now be used, in the spirit of \cite{BaeVal10a}, to construct a geometric invariant measuring the
non-Kerrness of the initial data on the compact set
$\mathcal{U}$. More precisely, we define
\begin{align}
\label{DefinitionNonKerrness}
I \equiv{}&  \int_{\mathcal{U}} \nabla_{(AB} \kappa_{CD)}
\widehat{\nabla^{AB} \kappa^{CD}} \mbox{d}\mu \nonumber\\
&+ \int_{\mathcal{U}} \Psi_{(ABC}{}^P \kappa_{D)P}
\widehat{\Psi^{ABCQ} \kappa^D{}_Q} \mbox{d}\mu. 
\end{align}

\subsection{The main result}
The main result of our analysis is the following theorem:

\begin{theorem}\label{MainTheorem}
Let $(\mathcal{S},h_{ij},K_{ij})$ be an initial data set for the
Einstein vacuum field equations, and let $\mathcal{U}\subset
\mathcal{S}$ be a compact connected subset with boundary $\partial
\mathcal{U}$ satisfying Assumption \ref{AssumptionXi}. Let $I$ be as
defined by equation \eqref{DefinitionNonKerrness} where $\kappa_{AB}$
is given as the only solution to equation
\eqref{ApproximateKillingSpinorEquation} with boundary behaviour given
by the $n_{AB}$-Killing spinor candidate $\breve{\kappa}_{AB}$ where
$n_{AB}$ is the inward pointing normal to $\partial \mathcal{U}$. If:
\begin{itemize}
\item[(i)] $I=0$; 

\item[(ii)] there exists a point on $\mathcal{U}$ for which
\begin{align}
&\mbox{\em Im}(\xi)=0, \quad \mbox{\em Im}(\xi_{AB})=0,\nonumber\\
&D_{AB}\mbox{\em Im} (\xi)=0, 
\quad D_{(A}{}^{C} \mbox{\em Im}(\xi_{B)C})=0;
\end{align}
\end{itemize}
 then the future domain of dependence, $D^+(\mathcal{U})$, of $\mathcal{U}$ is locally isometric to a subset of
a generalised Kerr-NUT spacetime. If, in addition:
\begin{itemize}
\item[(iii)] there exists a point on $\mathcal{U}$ for which $\Phi \neq 0$;

\item[(iv)] there exists a point on $\mathcal{U}$ for which
\begin{equation}
\lambda + \tfrac{3}{4} \mbox{\em Re}(\psi^{1/3})>0,
\label{PointwiseCondition}
\end{equation}
\end{itemize}
then $D^+(\mathcal{U})$ is locally isometric to a portion of a Kerr spacetime.
Conversely, on a compact subset $\mathcal{U}\subset \mathcal{S} $ of a
Kerr initial data set, $(\mathcal{S},h_{ij},K_{ij})$, the properties (i),
(ii), (iii) and (iv) are satisfied.
\end{theorem}

\medskip
\noindent
\textbf{Remark 1.} If $D^+(\mathcal{U})$ is locally isometric to a portion of a
Kerr spacetime, the conditions {\it (ii)}, {\it (iii)} and {\it (iv)}
are satisfied on every point. Hence, the choice of which point to
check the conditions in, is not important.

\medskip
\noindent
\textbf{Remark 2.} If $\mathcal{U}$ is not connected, the conditions {\it (ii)}, {\it (iii)} and {\it (iv)} needs to be checked for each connected component of $\mathcal{U}$.

\medskip
\noindent
\textbf{Remark 3.} The conditions {\it (iii)} and {\it (iv)} can be replaced by an asymptotic flatness condition.

\begin{proof}
If $I=0$ then it follows from our smoothness assumptions that
equations \eqref{kspd1}-\eqref{kspd2} are satisfied on
$\mathcal{U}$. Hence, from Theorem~\ref{Theorem:kspd} it follows that
$D^+(\mathcal{U})$ will contain a Killing spinor $\kappa_{AB}$. Then
$\xi_{AA'}$ is the spinor counterpart of a (possibly complex) Killing
vector. Now, using assumption {\it(ii)} together with
Lemma~\ref{Lemma:KDTriviality} gives $\mbox{Im}(\xi_{AA'})=0$ and
$\nabla_{AA'}\mbox{Im}(\xi_{BB'})=0$ at a point. Using a standard
result about Killing spinors (see Appendix C.3 in \cite{Wal84}), one
concludes that $\mbox{Im}(\xi)=\mbox{Im}(\xi_{AB})=0$ everywhere on
$D^+(\mathcal{U})$ so that $\xi_{AA'}$ is, in fact, real. Thus,
$D^+(\mathcal{U})$ is locally isometric to a portion of a generalised Kerr-NUT
spacetime.

\smallskip
 As in the main text, let $\phi_{AB}$ denote the spinorial counterpart
of the Killing form for of $\xi_{AA'}$. From the discussion in
Subsection \ref{Subsection:PropertiesKillingSpinors} one concludes
that
\[
\Psi_{ABCD} =\varpi \phi_{(AB} \phi_{CD)},
\]
for some function $\varpi$. Now, if $\Phi\neq 0$ on $\mathcal{U}$,
then using Theorem~\ref{Theorem:Mars2000}, one has that 
\[
\varpi = -\frac{12}{\tilde{c}-\sigma}, \qquad \Phi = -k (\tilde{c}-\sigma)^4,
\]
for some (possibly complex) constants $\tilde{c}$ and $k$. Using
formulae \eqref{ErnstPotential:Psi} and \eqref{TypeD:Phi}, one can
identify the constants $c$ and $\tilde{c}$ and set $k=\tfrac{8}{81}$. Evaluating $c$ at the
point where \eqref{PointwiseCondition} holds one obtains that
$\mbox{Re}(c)>0$. Thus, the hypothesis of Theorem~\ref{Theorem:Mars2000} hold and one concludes that $D^+(\mathcal{U})$
is locally isometric to a portion of the Kerr spacetime.

\medskip
Now, given a compact subset $\mathcal{U}\subset \mathcal{S} $ of a
Kerr initial data set, $(\mathcal{S},h_{ij},K_{ij})$, one knows there
exist a spinor $\kappa_{AB}$ for which the spatial  Killing spinor equations
\eqref{kspd1}-\eqref{kspd2} are satisfied. This spinor coincides at
$\partial\mathcal{U}$ (up to an irrelevant constant numerical factor)
with the $n_{AB}$-Killing spinor candidate. Thus, by uniqueness of the
elliptic problem \eqref{ApproximateKillingSpinorEquation} the
approximate Killing spinor obtained from solving the equation and
$\kappa_{AB}$ coincide (again, up to an irrelevant numerical factor)
and one has $I=0$ and (i) is satisfied. As $\kappa_{AB}$ satisfies the spatial Killing
spinor equations, it follows from the general theory of
\cite{BaeVal10b} that $(\xi,\xi_{AB})$ is a Killing vector initial
data set (KID). For Kerr this data corresponds to the real stationary Killing vector, thus (ii) is satisfied. 
Now, as $\psi\neq 0$ for the Kerr spacetime, one has  from equation
\eqref{TypeD:Phi} that $\Phi \neq 0$ and thus (iii) holds. Finally,
an explicit computation with the Kerr spacetime shows that
\eqref{PointwiseCondition} holds for any point of the Kerr spacetime
---hence one obtains (iv).
\end{proof}

\section{Conclusions and discussion}
\label{Section:Conclusions}
In this paper we have devised a way to measure the deviation from Kerr
initial data for bounded domains. The main result is presented in
Theorem~\ref{MainTheorem}. In the previous papers \cite{BaeVal10a,
BaeVal10b, BaeVal11b}, a similar result was obtained for cases where
the computational domain reached spatial infinity. For such cases the
asymptotic behaviour of the approximate Killing spinor could be
specified in a way that helped us to exclude all other Petrov type D
solutions. Therefore we could conclude that the data was Kerr data if
and only if $I=0$.  As the present paper deals with bounded domains,
we constructed the boundary data for the approximate Killing spinor
from the curvature. The drawback is that this gives $I=0$ for all type
D solutions. Therefore, one requires conditions {\it (ii)}, {\it
(iii)}, {\it (iv)} in Theorem~\ref{MainTheorem} to single out the Kerr
solution. An effort was put into formulating the conditions so they can be
verified at a single arbitrarily chosen point of the computational
domain. Furthermore, we have shown that a part of the invariant
constructed in \cite{BaeVal10a, BaeVal10b, BaeVal11b} can be omitted
in the case of a bounded domain as well the unbounded case.

\medskip
The results of this paper can be used to numerically evaluate how much
any slice of a spacetime deviates from Kerr data. This gives a tool to
quantify decay towards Kerr data for a numerically evolved
spacetime. A project along these lines have been initiated.

\section*{Acknowledgments}
Part of this research was carried out at the Erwin Schr\"odinger
Institute of the University of Vienna, Austria, during the course of
the programme ``Dynamics of General Relativity: Numerical and
Analytical Approaches'' (July-September, 2011). The authors thank the
organisers for the invitation to attend this programme and the
institute for its hospitality. We have profited from interesting
discussions with Dr. M. Mars. TB is funded by the Max-Planck Institute
for Gravitational Physics, Albert Einstein Institut.

\appendix

\section{Redundancy of the second algebraic condition}
\label{RedundancySecondAlgebraicCondition}

The purpose of the present appendix is to prove the assertion made in
Theorem~\ref{Theorem:kspd} that the second algebraic condition
given by equation \eqref{kspd3} is a consequence of the conditions
\eqref{kspd1} and \eqref{kspd2}. As a consequence of this result, the
conditions required on an initial data set to have a development with a
valence 2 Killing spinor become completely analogue to those required
to have a valence 1 Killing spinor ---see e.g. \cite{BaeVal11a}.

\medskip
The analysis in this appendix proceeds by discussing the various
possible algebraic types that the spinor $\kappa_{AB}$ can have. Our
first result is the following:

\begin{lemma}\label{lemmaD}
Assume that the symmetric spinor $\kappa_{AB}$ satisfies 
\begin{align*}
&\kappa_{AB}\kappa^{AB}\neq 0, \quad \nabla_{(AB}\kappa_{CD)}=0, \\
&\Psi_{(ABC}{}^F\kappa_{D)F}=0,
\end{align*}
 on an open subset $\mathcal{U}\subset \mathcal{S}$. Then the 
 algebraic condition \eqref{kspd3} is satisfied on $\mathcal{U}$.
\end{lemma}

\begin{proof}
The condition $\kappa_{AB}\kappa^{AB}\neq 0$ allows us to choose a
spin dyad $(o_A, \iota_A)$ and a scalar field $\varkappa$ such that
$o_A\iota^A=1$ and $\kappa_{AB}=e^\varkappa
o_{(A}\iota_{B)}$. Similarly, the condition
$\Psi_{(ABC}{}^F\kappa_{D)F}=0$ implies that there is a scalar field
$\psi$ such that $\Psi_{ABCD}=\psi o_{(A}o_B\iota_C\iota_{D)}$.

\medskip
In the next step we decompose the equation
$\nabla_{(AB}\kappa_{CD)}=0$ into its various components to obtain:
\begin{subequations}
\begin{eqnarray}
&& o^{A} o^{B} o^{C} \nabla_{AB} o_{C} = 0,\label{kappaeqs1} \\
&& o^{A} \iota^{B} o^{C} \nabla_{AB} o_{C} = -\tfrac{1}{2} o^{A} o^{B} \nabla_{AB}\varkappa,\label{kappaeqs2}\\
&& o^{A} o^{B} \iota^{C} \nabla_{AB}\iota_{C} 
- \iota^{A}\iota^{B}o^{C} \nabla_{AB} o_{C}\nonumber\\&&{}\hspace{2.5cm}= 2 o^{A} \iota^{B} \nabla_{AB}\varkappa,\label{kappaeqs3}\\
&& o^{A}\iota^{B}\iota^{C} \nabla_{AB}\iota_{C} = \tfrac{1}{2}\iota^{A} \iota^{B} \nabla_{AB}\varkappa,\label{kappaeqs4}\\
&& \iota^{A} \iota^{B} \iota^{C} \nabla_{AB} \iota_{C} = 0. \label{kappaeqs5}
\end{eqnarray}
\end{subequations}
These equations imply, in turn, that 
\begin{eqnarray}
&& e^{-\varkappa} \xi_{AB} ={}
-3 o_{A} o_{B} o^{C} \iota^{D} \iota^{F} \nabla_{CD}\iota_{F} \nonumber\\ 
&&\hspace{2cm} - 3 \iota_{A} \iota_{B}  o^{C}\iota^{D} o^{F}\nabla_{CD} o_{F}\nonumber\\ 
&& \hspace{2cm} + \tfrac{3}{2} o_{(A}\iota_{B)} (o^{C} o^{D} \iota^{F} \nabla_{CD}\iota_{F}\nonumber\\
&& \hspace{2.8cm} + \iota^{C}\iota^{D}o^{F} \nabla_{CD} o_{F})\label{xi2eq}.
\end{eqnarray}

Now, it is well known that the spacetime Bianchi identity $\nabla^Q{}_{A'}\Psi_{ABCQ}=0$ implies the constraint 
\begin{equation}\label{Bianchi3Da}
\nabla^{CD}\Psi_{ABCD}=0,
\end{equation}
on $\mathcal{S}$. Substituting $\Psi_{ABCD}=\psi
o_{(A}o_B\iota_C\iota_{D)}$ and  contracting with combinations of $o^A$ and $\iota^A$ one finds that the content of \eqref{Bianchi3Da} is given by 
\begin{subequations}
\begin{eqnarray}
&& o^A o^B\nabla_{AB}\psi ={} 
6 \psi o^A\iota^B o^C\nabla_{AB}o_C, \label{Bianchi3Db1}\\
&& o^{B} \iota^{C} \nabla_{BC}\psi ={} 
  \tfrac{3}{2}\psi \iota^{A}\iota^{B}o^{C}\nabla_{AB}o_{C}\nonumber\\ 
&&\hspace{2cm}- \tfrac{3}{2}\psi o^{A} o^{B} \iota^{C}\nabla_{AB}\iota_{C}, \label{Bianchi3Db2}\\
&& \iota^A\iota^B\nabla_{AB}\psi ={} 
-6\psi o^A\iota^B\iota^C\nabla_{AB}\iota_C. \label{Bianchi3Db3}
\end{eqnarray}
\end{subequations}
Using equation \eqref{xi2eq} and the Bianchi identities \eqref{Bianchi3Db1}-\eqref{Bianchi3Db3} we get
\begin{align}
&\Psi_{(ABC}{}^{F}\xi_{D)F} + 3 \kappa_{(A}{}^{F}\nabla_{B}{}^{H}\Psi_{CD)FH}\nonumber\\
&\hspace{1cm}=\tfrac{3}{4} e^{\varkappa}\psi \iota_{A} \iota_{B} \iota_{C} \iota_{D} 
o^{M} o^{P} o^{Q}\nabla_{PQ}o_{M} \nonumber\\
&\hspace{1.5cm}- \tfrac{3}{4} e^{\varkappa}\psi o_{A} o_{B} o_{C} o_{D} \iota^{M} \iota^{P} \iota^{Q}\nabla_{PQ}\iota_{M}. \nonumber
\end{align}
Finally using the information about the derivatives of the spin dyad
contained in equations \eqref{kappaeqs1}-\eqref{kappaeqs5} one finds
that we get that the second algebraic condition, equation
\eqref{kspd3}, is satisfied on $\mathcal{U}$. Notice that in this
argument one could have had $\psi=0$.
\end{proof}

Using similar methods as before, one obtains the following lemma:

\begin{lemma}\label{lemmaN}
Assume that the symmetric spinor $\kappa_{AB}$ satisfies 
\begin{align*}
&\kappa_{AB}\kappa^{AB}= 0, \quad \kappa_{AB}\hat\kappa^{AB}\neq 0, \\
&\nabla_{(AB}\kappa_{CD)}=0,\quad \Psi_{(ABC}{}^F\kappa_{D)F}=0, 
\end{align*}
on an open subset $\mathcal{U}\subset \mathcal{S}$. Then the algebraic
condition \eqref{kspd3} is satisfied on $\mathcal{U}$.
\end{lemma}

\begin{proof}
By assumption the $\kappa_{AB}$ is algebraically special ---that is,
it has repeated principal spinors. Thus, there exists $o_A$ such that
$\kappa_{AB}=o_A o_B$. We then complete $o_A$ to a normalised spinor
dyad $(o_A,\iota_A)$. The equation $\nabla_{(AB}\kappa_{CD)}=0$ is
equivalent to
\begin{subequations}
\begin{eqnarray}
&& o^{A} o^{B} o^{C} \nabla_{(AB} o_{C)}=0,\label{kappaeqsN1}\\
&& o^{A} o^{B} \iota^{C} \nabla_{(AB}o_{C)} = 0,\label{kappaeqsN2}\\
&& o^{A} \iota^{B} \iota^{C} \nabla_{(AB}o_{C)} = 0,\label{kappaeqsN3}\\
&& \iota^{A} \iota^{B} \iota^{C} \nabla_{(AB}o_{C)} = 0.\label{kappaeqsN4}
\end{eqnarray}
\end{subequations}
These equations imply, in turn, that
\begin{align}\label{xi2eqN}
\xi_{AB} &= -2 o_{A} o_{B} \iota^{C} \nabla_{CD}o^{D} 
+ 2 o_{(A}\iota_{B)} o^{C} \nabla_{CD}o^{D}. 
\end{align}
The condition $\Psi_{(ABC}{}^F\kappa_{D)F}=0$ implies that there is a
scalar field $\psi$ such that $\Psi_{ABCD}=\psi o_{(A}o_B o_C o_{D)}$.
Using this together with \eqref{xi2eqN} yields
\begin{align}
&\Psi_{(ABC}{}^{F}\xi_{D)F} + 3 \kappa_{(A}{}^{F}\nabla_{B}{}^{H}\Psi_{CD)FH} \nonumber\\
&\hspace{1cm}=-3 o_{A} o_{B} o_{C} o_{D} \psi o^{P} o^{Q} \iota^{R} \nabla_{(PQ}o_{R)}\nonumber\\ 
& \hspace{1.5cm}+ 3 o_{(A}o_{B}o_{C}\iota_{D)} \psi o^{P} o^{Q} o^{R} \nabla_{(PQ}o_{R)}. 
\end{align}
Finally using the relations \eqref{kappaeqsN1}-\eqref{kappaeqsN4} we
get that the second algebraic condition, equation~\eqref{kspd3}, is
satisfied on $\mathcal{U}$.
\end{proof}

\medskip
With the aid of the previous two lemmas, one can provide a proof of
Theorem \ref{Theorem:kspd} in the main text.

\begin{proof}
Let $\mathcal{U}_1$ be the set of all points in $\mathcal{S}$ where
$\kappa_{AB}\kappa^{AB}\neq 0$ and $\mathcal{U}_2$ be the set of all
points in $\mathcal{S}$ where $\kappa_{AB}\hat\kappa^{AB}\neq 0$. The
scalar functions $\kappa_{AB}\kappa^{AB}: \mathcal{S}\rightarrow
\mathbb{C}$ and $\kappa_{AB}\hat\kappa^{AB}: \mathcal{S}\rightarrow
\mathbb{R}$ are continuous. Therefore, $\mathcal{U}_1$ and
$\mathcal{U}_2$ are open sets. Now, let $\mathcal{V}_1$ and
$\mathcal{V}_2$ denote, respectively, the interiors of
$\mathcal{S}\setminus \mathcal{U}_1$ and $\mathcal{V}_1\setminus
\mathcal{U}_2$. On the open set $\mathcal{V}_1 \cap \mathcal{U}_2$ we
have that $\kappa_{AB}\kappa^{AB}=0$ and
$\kappa_{AB}\hat\kappa^{AB}\neq 0$.  Hence, by Lemma~\ref{lemmaN} the
second algebraic condition, equation \eqref{kspd3}, is satisfied on
$\mathcal{V}_1 \cap \mathcal{U}_2$. Similarly, by Lemma~\ref{lemmaD}
the condition \eqref{kspd3} is satisfied on $\mathcal{U}_1$.  On the
open set $\mathcal{V}_2$ we have that $\kappa_{AB}=0$ and therefore
equation \eqref{kspd3} is trivially satisfied on
$\mathcal{V}_2$. Using the above sets, the 3-manifold $\mathcal{S}$
can be split as
\begin{equation*}
\text{int} \mathcal{S} 
\subset \mathcal{U}_1 \cup ( \mathcal{V}_1 \cap \mathcal{U}_2) \cup \mathcal{V}_2 \cup \partial \mathcal{U}_1 \cup \partial \mathcal{V}_2.
\end{equation*}
The left hand side of equation \eqref{kspd3} is continuous and vanishes 
on the open sets $\mathcal{U}_1$, $\mathcal{V}_1 \cap \mathcal{U}_2$ and $\mathcal{V}_2$. By continuity it therefore also vanishes on the boundaries $\partial \mathcal{U}_1$ and $\partial \mathcal{V}_2$.
We can therefore conclude that \eqref{kspd3} is satisfied everywhere on $\text{int} \mathcal{S}$. Again by continuity this extends to $\mathcal{S}$. Finally,
using Theorem~2 in \cite{BaeVal10b} one obtains the existence of a
valence-2 Killing spinor on $D^+(\mathcal{S})$.
\end{proof}



\end{document}